\documentclass[a4paper]{sig-alternate}
\usepackage[plainpages=false,pdfpagelabels,colorlinks=true,citecolor=blue,hypertexnames=false]{hyperref}
\usepackage{amssymb}
\usepackage{bbm}
\usepackage[utf8]{inputenc}
\usepackage{algo}

\def\<#1>{\langle#1\rangle}
\let\set\mathbbm
\def\vec#1{\mathbf{#1}}
\def\spread{\operatorname{Spread}}
\def\disp{\operatorname{Disp}}

\def\lcm{\operatorname{lcm}}

\def\K{\set K}

\newtheorem{theorem}{Theorem}

\newtheorem{lemma}{Lemma}
\newtheorem{definition}{Definition}
\newtheorem{example}{Example}

\newtheorem{algorithm}{Algorithm}

\begin{document}

\title{A Refined Denominator Bounding Algorithm for Multivariate Linear Difference Equations}

\numberofauthors{2}
\author{
 \alignauthor Manuel Kauers\titlenote{Supported by the Austrian FWF grant Y464-N18 and 
     the EU grant PITN-GA-2010-264564.}\\[\smallskipamount]
      \affaddr{RISC}\\
      \affaddr{Johannes Kepler University}\\
      \affaddr{4040 Linz (Austria)}\\[\smallskipamount]
      \email{mkauers@risc.jku.at}
 \alignauthor Carsten Schneider\titlenote{Supported by the Austrian FWF grant P20347-N18
     and the EU grant PITN-GA-2010-264564.}\\[\smallskipamount]
      \affaddr{RISC}\\
      \affaddr{Johannes Kepler University}\\
      \affaddr{4040 Linz (Austria)}\\[\smallskipamount]
      \email{cschneid@risc.jku.at}
}

\maketitle

\begin{abstract}
  We continue to investigate which polynomials can possibly occur as factors in the
  denominators of rational solutions of a given partial linear difference equation.
  In an earlier article we had introduced the distinction between periodic and
  aperiodic factors in the denominator, and we gave an algorithm for predicting
  the aperiodic ones. Now we extend this technique towards the periodic case
  and present a refined algorithm which also finds most of the periodic factors.
\end{abstract}

\kern-\medskipamount


\kern-\medskipamount


\kern-\medskipamount


\section{Introduction}

The usual approach for finding rational solutions of linear difference equations
with polynomial coefficients is as follows. First one constructs a nonzero polynomial~$Q$ such
that for any solution $y=p/q$ of the given equation we must have $q\mid Q$.
Such a polynomial~$Q$ is called a denominator bound for the equation.  Next, the
denominator bound is used to transform the given equation into a new equation
with the property that a polynomial~$P$ solves the new equation if and only if
the rational function $y=P/Q$ solves the original equation.  Thus the
knowledge of a denominator bound reduces rational solving to polynomial solving.

The first algorithm for finding a denominator bound~$Q$ was given by Abramov in
1971~\cite{abramov71,Abramov:89b,abramov91}. During the past fourty years, other
algorithms were found~\cite{paule95,hoeij98,bostan06,chen08b} and the technique
was generalized to matrix equations~\cite{abramov98a,barkatou99} as well as to
equation over function fields~\cite{Petkov:92,bronstein00,schneider04c}. Last
year~\cite{kauers10b} we made a first step towards a denominator bounding
algorithm for equations in several variables (PLDEs). We found that some factors
of the denominator are easier to predict than others. We called a polynomial
periodic if it has a nontrivial gcd with one of its shifts, and aperiodic
otherwise. For example, the polynomial $2n-3k$ is periodic because shifting it
twice in~$k$ and three times in~$n$ leaves it fixed.  We say that it is periodic
in direction $(3,2)$. An example for an aperiodic polynomial is $n k +1$. The
main result of last year's paper was an algorithm for determining aperiodic
denominator bounds for PLDEs, i.e., we can find $Q$ such that whenever
$y=\frac{p}{uq}$ solves the given equation and $q$ is aperiodic, then $q\mid Q$.

The present paper is a continuation of this work. We now turn to periodic
factors and study under which circumstances a slightly adapted version of last
year's algorithm can also predict periodic factors of the denominator. We
propose an algorithm which finds the periodic factors for almost all
directions. Every equation has however some directions which our algorithm does
not cover. But if, for instance, we have a system of two equations and apply our
algorithm to each of them, then the two bounds can under favorable circumstances
(which can be detected algorithmically) combined to a denominator bound which
provably contains all the factors that can possibly occur in the denominator of
any solution of the system. This was not possible before. So while until now we
were just able to compute in all situations some factors, we can now also find
in some situations all factors.

Despite this progress, we must confess that our results are still of a somewhat
academic nature because denominator bounds in which some factors are missing are
not really enough for solving equations. And even when a full denominator bound
is known, it still remains to find the polynomial solutions of a PLDE, and
nobody knows how to do this---the corresponding problem for differential
equations is undecidable. But in practice, we can heuristically choose a degree
bound for finding polynomial solutions, and knowing parts of the possible
denominators is certainly better than knowing nothing, and the more factors we
know, the better. Apart from this, we find it interesting to see how far the
classical univariate techniques carry in the multivariate setting, and we would
be curious to see new ideas leading towards algorithms which also find the
factors that we still miss.

\section{Preparations}

Let $\K$ be a field of characteristic zero. We consider polynomials and rational functions
in the $r$ variables $n_1,\dots,n_r$ with coefficients in~$\K$. For each variable~$n_i$,
let $N_i$ denote the shift operator mapping $n_i$ to $n_i+1$ and leaving all other variables
fixed, so that
\begin{alignat*}1
  &N_i q(n_1,\dots,n_r)\\
  &\qquad=q(n_1,\dots,n_{i-1},n_i+1,n_{i+1},\dots,n_r)
\end{alignat*}
for every rational function~$q$. Whenever it seems appropriate, we will use multiindex
notation, writing for instance $\vec n$ instead of $n_1,\dots,n_r$ or $N^{\vec i}$
for $N_1^{i_1}N_2^{i_2}\cdots N_r^{i_r}$.

We consider equations of the form
\begin{equation}\label{eq:main}
  \sum_{\vec s\in S}a_{\vec s}N^{\vec s}y=f
\end{equation}
where $S\subseteq\set Z^r$ is finite and nonempty, $f\in\K[\vec n]$ and $a_{\vec s}\in\K[\vec n]\setminus\{0\}$
($\vec s\in S$) are given, and $y$ is an unknown rational function.
Our goal is to determine the polynomials $p\in\K[\vec n]$ which may possibly occur in the denominator
of a solution~$y$, or at least to find many factors of~$p$.

We recall the following definitions and results from our previous paper~\cite{kauers10b}.

\begin{definition} Let $p,q,d\in\K[\vec n]$.
\begin{enumerate}
\item The set $\spread(p,q):=\{\,\vec i\in\set Z^r:\gcd(p,N^{\vec i}q)\neq1\,\}$ is called
  the \emph{spread} of $p$ and~$q$.
  For short, we write $\spread(p):=\spread(p,p)$.
\item The number $\disp_k(p,q):=\max\{\,|i_k|:(i_1,\dots,i_r)\in\spread(p,q)\,\}$ is called
  the \emph{dispersion} of $p$ and $q$ with respect to $k\in\{1,\dots,r\}$. (We set
  $\max A:=-\infty$ if $A$ is empty and $\max A:=\infty$ if $A$ is unbounded.)
\item The polynomial $p$ is called \emph{aperiodic} if $\spread(p)$ is finite, and \emph{aperiodic}
  otherwise.
\item The polynomial~$d$ is called an \emph{aperiodic denominator bound} for equation~\eqref{eq:main}
  if $d\neq0$ and every solution~$y$ can be written as $\frac{a}{ub}$ for some $a,b,u\in\K[\vec n]$
  where $u$ is periodic and $b\mid d$.
\item A point $\vec p\in S\subseteq\set Z^r\subseteq\set R^r$ is called a \emph{corner point} of $S$
  if there exists a vector $\vec v\in\set R^r$ such that $(\vec s-\vec p)\cdot\vec v>0$ for all
  $\vec s\in S\setminus\{\vec p\}$. Such a vector $\vec v$ is then called an \emph{inner vector,}
  and the affine hyperplane $H:=\{\vec x\in\set R^r:(\vec x-\vec p)\cdot\vec v=0\}$ is called a
  \emph{border plane} for~$S$.

  \medskip
  \centerline{\epsfig{file=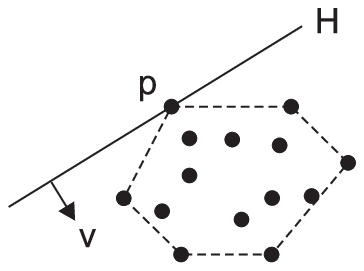}}

\end{enumerate}
\end{definition}

\begin{theorem}\label{thm:1} Let $p,q\in\K[\vec n]$.
\begin{enumerate}
\item If $p$ is irreducible, then $\spread(p)$ is a submodule of~$\set Z^r$ and
  $p$ is aperiodic if and only if $\spread(p)=\{0\}$.
\item\label{thm:1.2} If $p$ and $q$ are irreducible, then there exists $\vec s\in\set Z^r$ such that
  $\vec s+\spread(p,q)$ is a submodule of~$\set Z^r$.
\item There is an algorithm for computing $\spread(p,q)$.
\item There is an algorithm for computing an aperiodic denominator bound for~\eqref{eq:main}
  given the support $S$ and the coefficients $a_{\vec s}$ ($\vec s\in S$).
\end{enumerate}
\end{theorem}

\section{Denominator Bounds modulo a Prescribed Module}

Our goal in this section is to determine the factors whose spread is contained
in some prescribed set~$W\subseteq\set Z^r$. Under suitable assumptions
about~$W$ such factors must pop up in the coefficients of the equation (cf.\
Lemma~\ref{lem:2} below) and under stronger assumptions we can also give a bound
on the dispersion between them (cf.\ Theorem~\ref{thm:disp} below). Using these
two results we obtain a denominator bound relative to~$W$ (cf.\
Theorem~\ref{thm:db} and Algorithm~\ref{algo:1}) below. In the next section, we
then propose an algorithm which combines the denominator bounds with respect to
several sets~$W$. It turns out that by considering only finitely many sets~$W$ 
one can obtain a denominator bound with respect to infinitely many sets~$W$.

\begin{definition}
  Let $W\subseteq\set Z^r$ with $0\in W$.
  A polynomial $d\in\K[\vec n]\setminus\{0\}$ is called a denominator bound of \eqref{eq:main}
  with respect to~$W$ if for every solution $y=p/q\in\K(\vec n)$
  of \eqref{eq:main} and every irreducible factor $u$ of~$q$ with
  $\spread(u)\subseteq W$ we have $u\mid d$.
\end{definition}

Typically, $W$~will be a submodule of $\set Z^r$ or a union of such modules.
The definition reduces to the notion of aperiodic denominator bound
when $W=\{0\}$. In the other extreme, when $W=\set Z^r$ then $d$ is
a ``complete'' denominator bound: it contains all the factors, periodic or not,
that can possibly occur in the denominator of a solution~$y$ of~\eqref{eq:main}.
In general, $d$~predicts all aperiodic factors in the denominator of a solution
as well as the periodic factors whose spread is contained in~$W$.

Denominator bounds with respect to different submodules can be combined as follows.

\begin{lemma}\label{lemma:lcm}
  Let $W_1,\dots,W_m$ be submodules of~$\set Z^r$, and let
  $d_1,\dots,d_m$ be denominator bounds of \eqref{eq:main} with respect to
  $W_1,\dots,W_m$, respectively. Then
  $d:=\lcm(d_1,\dots,d_m)$ is a denominator bound with respect to
  $W:=W_1\cup\cdots\cup W_m$.
\end{lemma}

\begin{proof}
  Let $u$ be an irreducible factor of the denominator of some solution
  of~\eqref{eq:main} and suppose that $U:=\spread(u)\subseteq W$. It suffices
  to show that then $U\subseteq W_k$ for some~$k$, because then it follows that
  $u\mid d_k\mid d$, as desired.

  We show that if $U$ contains some vector $\vec x\not\in W_1$, then
  $W_1\subseteq W_2\cup\dots\cup W_m$, hence $U\subseteq W_2\cup\cdots\cup W_m$.
  Applying the argument repeatedly proves that $U\subseteq W_k$ for some~$k$.

  Let $\vec y\in W_1$. Since $U$ is a submodule of $\set Z^r$, we have $\vec
  x+\alpha\vec y\in U$ for all $\alpha\in\set Z$.  By assumption $U\subseteq
  W_1\cup\dots\cup W_m$, so each such $\vec x+\alpha\vec y$ must belong to at
  least one module~$W_\ell$ ($\ell=1,\dots,m$). It cannot belong to~$W_1$ though,
  because together with $\vec y\in W_1$ this would imply $\vec x\in W_1$, which
  is not the case. Therefore: For every $\alpha\in\set Z$ there exists
  $\ell\in\{2,\dots,m\}$ such that $\vec x+\alpha\vec y\in W_\ell$.

  Since $\set Z$ is infinite and $m$ is finite, there must be some index
  $\ell\in\{2,\dots,m\}$ for which there are two different
  $\alpha_1,\alpha_2\in\set Z$ with $\vec x+\alpha_1\vec y\in W_\ell$ and
  $\vec x+\alpha_2\vec y\in W_\ell$. Since $W_\ell$ is also a submodule of~$\set
  Z^r$, it follows that $(\alpha_1-\alpha_2)\vec y\in W_\ell$, and finally $\vec
  y\in W_\ell\subseteq W_2\cup\cdots\cup W_m$, as claimed.
\end{proof}

The next result says that factors of denominators tend to leave traces in the
coefficients of corner points of~$S$.

\begin{lemma}\label{lem:2}
  Let $W$ be a submodule of $\set Z^r$ and let $u$ with $\spread(u)\subseteq W$
  be an irreducible factor of the denominator of some solution~$y$ of~\eqref{eq:main}.
  Let $\vec p\in S$ be a corner point of $S$ with an inner vector $\vec v\in\set R^r$
  orthogonal to $W$ (meaning $\vec w\cdot\vec v=0$ for all $\vec w\in W$).
  Then there exists $\vec i\in\set Z^r$ such that $N^{\vec i}u\mid a_{\vec p}$.
\end{lemma}

\begin{proof}
  If $u'$ is another irreducible factor of the denominator of~$y$ and
  $\spread(u,u')$ is nonempty, then we have $\spread(u,u')\subseteq\vec c+W$ for some $\vec c\in\set Z^r$.
  This follows from Theorem~\ref{thm:1}.\ref{thm:1.2} and the assumption $\spread(u)\subseteq W$.
  For the full denominator~$d$ of~$y$, we can thus find $\vec c_1,\dots,\vec c_m\in\set Z^r$
  with $\spread(u,d)\subseteq\bigcup_{k=1}^m(\vec c_k+W)$ where each element
  from $C$ is necessary.
  Let $C=\{\vec c_1,\dots,\vec c_m\}$ be such a choice, and let $i\in\{1,\dots,m\}$ be such that
  $\vec c_i\cdot\vec v$ is minimal.

  We have
  \[
    a_{\vec p}N^{\vec p}y=f-\sum_{\vec s\in S\setminus\{\vec p\}}a_{\vec s}N^{\vec s}y
  \]
  as an identity in $\set K(\vec n)$. Therefore, every factor in the denominator of $N^{\vec p}y$
  must either be canceled by $a_{\vec p}$ or it also occurs as a factor in at least
  one of the $N^{\vec s}y$ ($\vec s\in S\setminus\{\vec p\}$).
  The factor $N^{\vec p+\vec c_i}u$ appears in the denominator of~$N^{\vec p}y$.
  If it also appeared in the denominator of $N^{\vec s}y$ for some $\vec s\in S\setminus\{\vec p\}$,
  then this would imply $N^{\vec p+\vec c_i}u=N^{\vec s+\vec c_j}u$ for some $j\in\{1,\dots,m\}$.
  But then $\vec s-\vec p + \vec c_j-\vec c_i\in W$, which is in contradiction to
  \begin{alignat*}1
    &(\vec s-\vec p + \vec c_j-\vec c_i)\cdot\vec v
    =\underbrace{(\vec s-\vec p)\cdot\vec v}_{>0} + \underbrace{(\vec c_j-\vec c_i)\cdot\vec v}_{\geq0}
     \neq0
  \end{alignat*}
  because $W$ is orthogonal to $\vec v$ by assumption.
  Hence $N^{\vec p+\vec c_i}u$ cannot appear as a denominator on the right hand side,
  and hence it must be canceled on the left hand side.
  This forces $N^{\vec p+\vec c_i}u\mid a_{\vec p}$, so the claim is proven for
  $\vec i:=\vec p+\vec c_i$. \qed
\end{proof}

The lemma tells us for which choices of $W\subseteq\set Z^r$ something nontrivial may happen.
Let us illustrate this with an example.

\begin{example}\label{ex:1}
  The equation
  \begin{alignat*}1
    &(4 k-2 n+1) (k+n+1) y(n,k)\\
    &\quad+(8 k^2+2 k n+k+6 n^2+13 n+6) y(n,k+1)\\
    &\qquad-2 (6 k^2+2 k n+13 k+2 n^2+n+6) y(n+1,k)=0
  \end{alignat*}
  has the solution $y=(n^2+2k^2)/(k+n+1)$. Its denominator is periodic, $\spread(k+n+1)=\binom{1}{-1}\set Z$.
  Lemma~\ref{lem:2} predicts the appearance of $k+n+1$ (or at least some shifted version of it)
  in the coefficient of~$y(n,k)$, because for the choice $W=\binom{1}{-1}\set Z$, the point
  $\vec p=\binom 00\in S$ admits the choice $\vec v=\binom11$ in accordance with the requirements imposed
  by the lemma. Note that no shift equivalent copy of $k+n+1$ occurs in the coefficients of $y(n,k+1)$
  or $y(n+1,k)$, which does not contradict the lemma, because the points $\binom01$ and $\binom10$ lie on
  a line parallel to~$W$. This has the consequence that for these points, there does not exist
  a vector $\vec v$ with the required property.

  Conversely, the factor $4k-2n+1$ cannot possibly appear in the denominator of a solution,
  because for $W':=\spread(4k-2n+1)=\binom12\set Z$ we can take $\vec p'=\binom10$ and
  $\vec v'=\binom{-1}{1/2}$, and according to the lemma, some shifted version of $4k-2n+1$
  would have to appear in the coefficient of $y(n+1,k)$.

  More generally, for any nontrivial submodule $W''$ of $\set Z^2$ other than~$W$, Lemma~\ref{lem:2}
  excludes the possibility of periodic factors whose spread is contained in~$W''$, because such
  factors would have to leave a trace in at least one of the coefficients of the equation.

  \medskip
  \centerline{\epsfig{file=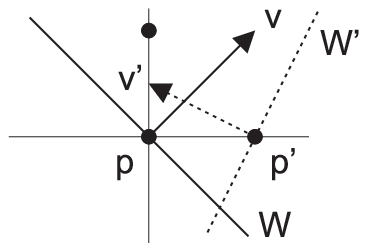}}

\end{example}

In the previous example, we could thus determine all the interesting modules $W$ by just looking
at the spreads of the the factors of the coefficients of the equation. The following example
indicates that this is not always sufficient.

\begin{example}\label{ex:2}
  The equation
  \begin{alignat*}1
    &(2 k-3 n^2-8 n-5) y(n,k+1)\\
    &\quad+(k+3 n^2+5 n+4) y(n+1,k)\\
    &\qquad-(5 k-3n^2-11 n-7) y(n+1,k+1)\\
    &\quad\qquad+(2 k-3 n^2-8 n-3) y(n+2,k)=0
  \end{alignat*}
  also has the solution $y=(n^2+2k^2)/(k+n+1)$. Its denominator $k+n+1$ does not appear in
  any of the coefficients of the equation. This is because for its spread $W=\binom{1}{-1}\set Z$
  there are no suitable $\vec p\in S$ and $\vec v\in\set R^2$ matching the conditions of the
  lemma because the points $\binom10,\binom01$ as well as the points $\binom20,\binom11$ lie
  on a line parallel to~$S$.
\end{example}

In summary, in order for $W$ to be the spread of a factor that can appear in the denominator
of a solution of~\eqref{eq:main}, $W$~must be contained in the spread of some coefficient
of the equation (as in Ex.~\ref{ex:1}) or it must be parallel to one of the faces in the
convex hull of the support~$S$ (as in Ex.~\ref{ex:2}).
For every equation, we can thus determine some finitely many submodules of $\set Z^r$ of
codimension one such that each possibly occuring spread $W$ is contained in at least one
of them.

\subsection{A Normalizing Change of Variables}\label{sec:normal}

Let $\set Z^r=V\oplus W$ be a decomposition of~$\set Z^r$ into submodules.
Our goal is to obtain denominator bounds with respect to~$W$ by applying the
algorithm from last year~\cite{kauers10b} to $V\cong\set Z^r/W$.
It turns out that this can be done provided
that $W$ is sufficiently nondegenerate. In order to formulate the precise
conditions on~$W$ without too much notational overhead, it seems convenient
to make a change of coordinates.

Let invertible matrices $A=((a_{i,j}))_{i,j=1}^r\in\set Q^{r\times r}$
act on $\K(\vec n)$ via
\begin{alignat*}3
  A\cdot y(n_1,\dots,n_r)
  &:= y\bigl(&&a_{1,1}n_1+a_{1,2}n_2+\cdots+a_{1,r}n_r,\\
  &          &&a_{2,1}n_1+a_{2,2}n_2+\cdots+a_{2,r}n_r,\\
  &          &&\vdots\\
  &          &&a_{r,1}n_1+a_{r,2}n_2+\cdots+a_{r,r}n_r\bigr).
\end{alignat*}
We obviously have $A\cdot(p+q)=(A\cdot p)+(A\cdot q)$ and
$A\cdot(pq)=(A\cdot p)(A\cdot q)$ for all $p,q\in\K(\vec n)$.
It can be checked that we also have
\[
  A\cdot(N^{\vec s}y)=N^{A^{-1}\vec s}(A\cdot y)
\]
for every $A\in\set Z^{r\times r}$ with $|\det A|=1$ and
every $\vec s\in\set Z^r$ and every $y\in\set K(\vec n)$.
It follows that $y\in\K(\vec n)$ is a solution of \eqref{eq:main}
if and only if $\tilde y=A^{-1}\cdot y$ is a solution of the transformed equation
\[
  \sum_{\vec s\in S}(A\cdot a_{\vec s})N^{A^{-1}\vec s}\tilde y=A\cdot f,
\]
or equivalently of
\[
  \sum_{\vec s\in\tilde S}\tilde a_{\vec s}N^{\vec s}\tilde y=\tilde f,
\]
where $\tilde S=\{A^{-1}\vec s:s\in S\}$, $\tilde a_{\vec s}:=A\cdot a_{A\vec s}$
($\vec s\in\tilde S$), and $\tilde f=A\cdot f$.

Now take $A\in\set Z^r$ with $|\det A|=1$ such that the first $t$ rows of $A$
form a basis of $V$ and the last $r-t$ rows of $A$ form
a basis of~$W$. Then the transformation just described maps
the basis vectors of $V$ to the first $t$ unit vectors and the basis
vectors of $W$ to the last $r-t$ unit vectors. In other words, we can
assume without loss of generality that $V$ itself is generated by the
first $t$ unit vectors and $W$ by the last $r-t$ unit vectors in~$\set Z^r$.
We will make this assumption from now on, unless otherwise stated.
Note that this convention implies for an irreducible polynomial $u\in\K[\vec n]$
that $\spread(u)=W$ is equivalent to $u$ being free of the variables
$n_{t+1},n_{t+2},\dots,n_r$ and aperiodic as element of $\K[n_1,\dots,n_t]$.

By applying, if necessary, a suitable power of $N_1$ on both sides of the
equation we can further assume without loss of generality that
$\min\{s_1:(s_1,\dots,s_r)\in S\}=0$, and we set
$k:=\max\{s_1:(s_1,\dots,s_r)\in S\}$.
\subsection{Bounding the Dispersion}

With this transformation w.r.t.\ the submodule $W$ of $\set Z^r$ and under the assumption that the extreme points~\eqref{Equ:ABSet} of $S$ have certain properties, Theorem~\ref{thm:disp} explains how one can bound the dispersion along the $x_1$-coordinate of all factors $f$ with $\spread(f)\in W$ that occur in the denominator of a solution of~\eqref{eq:main}. This result is a refinement of Lemma~2 from~\cite{kauers10b}.

\begin{lemma}\label{Lemma:UniqueTComponents}
Let $u,v\in\set K[\vec n]\setminus\{0\}$ with $\spread(u)\subseteq W$ and $\spread(v)\subseteq  W$. Then
\begin{multline*}
|\{(s_1,\dots,s_t)\in\set Z^t\mid\exists (s_{t+1},\dots,s_r)\in\set Z^{r-t}:\\
N^{(s_1,\dots,s_r)}u=v\}|\leq 1.
\end{multline*}
\end{lemma}
\begin{proof}
Take $\vec s,\vec s'$ with $N^{\vec s}u=v=N^{\vec s'}u$. As $N^{\vec s-\vec s'}u=u$, it follows $\vec{s}-\vec{s}'\in W=\{0\}^t\times\set Z^{r-t}$, and thus the first $t$ components of $\vec s,\vec s$ agree.
\end{proof}

\begin{theorem}\label{thm:disp}
  Let
  \begin{equation}\label{Equ:ABSet}
  \begin{split}
    A&=\{(s_1,\dots,s_r)\in S: s_1=0\},\\
    B&=\{(s_1,\dots,s_r)\in S: s_1=k\}.
    \end{split}
  \end{equation}
  Suppose that no two elements of $A$ agree in the first $t$ coordinates,
  and that the same is true for~$B$. Let $a'_{\vec i}$ be those polynomials which contain all irreducible factors $f$ of $a_{\vec i}$ with $\spread(f)\subseteq W$. Let
  \[
    s:=\max\{\disp_1(a'_{\vec s},N_1^{-k}a'_{\vec t}): \vec s\in A\text{ and }\vec t\in B\}.
  \]
  Then for any solution $y=p/q\in\K(\vec n)$ of $\eqref{eq:main}$ and any irreducible
  factors $u,v$ of $q$ with $\spread(u),\spread(v)\subseteq W$ we have $\disp_1(u,v)\leq s$.
\end{theorem}
\begin{proof}
As $S$ is not empty, $A,B$~are nonempty. W.l.o.g.\ we may assume that the minimal element of $A$ w.r.t.\ lexicographic order is the zero vector.\\
Suppose that there are irreducible factors $u,v$ of $q$ with $\spread(u)\subseteq W$, $\spread(v)\subseteq W$ and $d:=\disp_1(u,v)$ such that $d>s$; take such $u,v$ such that $d$ is maximal.
Consider all the factors $N^{\vec u}u$ and $N^{\vec v}v$ occurring in $q$ where the first entry in $\vec u$ and $\vec v$ is $0$. Note that by Lemma~\ref{Lemma:UniqueTComponents} there are only finitely many choices of the first $t$ components, so we can choose two such factors from $q$ where the first $t$ components of $\vec u$ are minimal and the first $t$ components of $\vec v$ are maximal w.r.t.\ lexicographic order; these factors are denoted by $u',v'$ respectively.\\
$\bullet$ First suppose that $u'$ divides one of the polynomials $a_{\vec{s}}$ with $\vec{s}\in A$.
In this case we choose the polynomial $a_{\vec w}$ with $\vec{w}=(w_1,\dots,w_r)\in B$ such that $(w_2,\dots,w_{t})$ is maximal w.r.t.\ lexicographic order (uniqueness is guaranteed by the assumption that no two elements from $B$ agree in the first $t$ components). We can write~\eqref{eq:main} in the form
\begin{equation}\label{Equ:RewritePLDE}
N^{\vec w}y=\frac{1}{a_{\vec{w}}}\Big(f-\sum_{\vec{s}\in S\setminus\{\vec{w}\}}a_{\vec{s}}N^{\vec{s}}y\Big).
\end{equation}
Now observe that the factor $N^{\vec w}v'$ does not occur in the denominator of any $N^{\vec s}y$ with $\vec s\in S\setminus\{\vec{w}\}$:
\begin{enumerate}
\item Suppose that there is $\vec{s}\in S\setminus B$ such that $N^{\vec w}v'$ occurs in $N^{\vec s}q$, i.e., $N^{\vec w-\vec s}v'$ is a factor of $q$. Since the first component of $\vec{w}$ is $k$ ($\vec w\in B$) and the first component of $\vec{s}$ is smaller than $k$ ($\vec{s}\notin B$), the first component of $\vec w-\vec s$ is positive. Moreover, since the distance between the factors $v'$ and $u'$ of $q$ is $d$ in the first component, the factors $v'$ and $N^{\vec w-\vec s}v'$ of $q$ have distance larger than $d$ in the first component; a contradiction that the distance $d$ is maximally chosen. Consequently, if $N^{\vec w}v'$ is a factor
    in the denominator of $N^{\vec s}y$ with $\vec{s}\in S$, it follows that $\vec{s}\in B$.
\item Suppose that there is $\vec{s}\in B$ with $\vec w\neq\vec{s}$ such that $N^{\vec w}v'$ is a factor of $N^{\vec s}q$. Then $N^{\vec w-\vec s}v'$ is a factor of $q$. Since the first component of the vectors in $B$ is $k$, but the first $t$ components in total cannot be the same for two different vectors of $B$, it follows that the first entry in $\vec w-\vec s$ is zero and at least one of the others is non-zero; in particular, by the maximality assumption on $\vec{w}$ the first non-zero entry is positive. Hence we find $\vec v'=(0,v'_2,\dots,v'_r):=\vec v+\vec w-\vec s$ such that $N^{\vec v'}v$ is a factor of $q$ and such that $(v'_2,\dots,v'_t)$ is larger than $(v_2,\dots,v_t)$ w.r.t.\ lexicographic ordering; a contradiction to the choice of the vector $\vec v$.
\end{enumerate}
Since $f,a_{\vec{s}}\in\set K[\vec{n}]$, the common denominator of the rational function
$f-\sum_{\vec{s}\in S\setminus\{\vec{w}\}}a_{\vec{s}}N^{\vec{s}}y$ does not contain the factor  $N^{\vec w}v'$.  Now suppose that $N^{\vec w}v'$ is a factor of $a_{\vec{w}}$. Since $\vec{w}\in B$,
its first component is $k$. But then, since $u'$ and $v'$ have distance $d$ in the first coordinate, also the factors $u'$ and $N_1^{-k}N^{\vec{w}}v'$ have distance $d$. Thus
$\disp_1(a_{\vec{s}},N_1^{-k} a_{\vec w})\geq d$ which implies that $s\geq d$; a contradiction. Overall, the common denominator on the right hand side of~\eqref{Equ:RewritePLDE} cannot contain the factor $N^{\vec w}v'$ which implies that the denominator of $N^{\vec w}y$ is not divisible by $N^{\vec w}v'$. Thus the denominator of $y$, in particular $q$ is not divisible by $v'$; a contradiction.\\
$\bullet$ Conversely, suppose that $u'$ does not divide any of the polynomials $a_\vec{s}$ with $\vec{s}\in A$. Now let $\vec w=(0,w_2\dots,w_r)\in A$ such that $(w_2,\dots,w_t)$ is minimal w.r.t.\ lexicographic ordering (again it is uniquely determined by the assumptions on $A$), and write~\eqref{eq:main} in the form~\eqref{Equ:RewritePLDE}; by our assumption stated in the beginning, $\vec w$ is just the zero vector $\vec{0}$. By analogous arguments as above (the roles of $A$ and $B$ are exchanged) it follows that $u'$ does not occur in the denominator of any $N^{\vec s}y$ with $\vec s\in S\setminus\{\vec{0}\}$. Hence as above, the common denominator of
$f-\sum_{\vec{s}\in S\setminus\{\vec{0}\}}a_{\vec{s}}N^{\vec{s}}y$ does not contain the factor  $u'$. Moreover, since $u'$ does not divide any $a_{\vec{s}}$ from $\vec{s}\in A$, the factor $u'$ does not occur in $a_{\vec{0}}$. In total, the factor $u'$ is not part of the denominator on the right hand side of~\eqref{Equ:RewritePLDE}, but it is a factor of the denominator on the left hand side; a
contradiction.
\end{proof}

If the required properties on the sets~\eqref{Equ:ABSet} in Theorem~\ref{thm:disp} are violated, our bounding strategy does not work, as can be seen by the following example.

\begin{example}
Fix $W:=\spread(k+n+1)=\binom{1}{-1}\set Z$ and take
$V=\binom{0}{1}\set Z$. The problem from Example~\ref{ex:1} is normalized
by the change of variables $n\to k$ and $k\to n-k$ (i.e., a basis transformation \tiny$\left(\begin{matrix}
0&1\\
1&-1
\end{matrix}\right)$\normalsize with determinant $-1$ is chosen) and one obtains $V'=\binom{1}{0}\set Z$ and $W'=\binom{0}{1}\set Z$. This gives the new equation
  \begin{alignat*}1
    &(n+1) (-6 k+4 n+1) y(n,k)\\
    &\quad+(12 k^2-14 n k+12
   k+8 n^2+n+6)y(n+1,k)\\
    &-2(6 k^2-10 n k-12 k+6n^2+13 n+6)y(n+1,k+1)=0
  \end{alignat*}
with the new structure set $S'=\{\binom{0}{0},\binom{1}{0},\binom{1}{1}\}$
which now has the solution $y=\frac{3 k^2-4 n k+2 n^2}{n+1}$ where the denominator consists of the factor $n+1$ with $\spread(n+1)=W'$. As observed already in Example~\ref{ex:1} one can predict the factor $n+1$ (up to a shift in $n$) by exploiting Lemma~\ref{lem:2}. However, one cannot apply Theorem~\ref{thm:disp}. For $S'$ we get the sets $A=\{\binom{0}{0}\}$ and $B=\{\binom{1}{0},\binom{1}{1}\}$ where in $B$ the two vectors are the same in the first component but differ in the second component.
\end{example}

\subsection{Denominator Bounding Theorem}

The denominator bounding theorem says that if we rewrite the equation \eqref{eq:main} into
a new equation whose support contains some point $\vec p$ which is sufficiently far away from
all the other points in the support, then we can read off a denominator bound from this
new equation. We will need the following fact, which appears literally as Theorem~3 in~\cite{kauers10b}
(with $W,S'$ renamed to $R^-,R^+$ here in order to avoid a name clash with the meaning of $W$
in the present paper).

\begin{lemma}\label{lemma:1}
  Let $\vec p$ be a corner point of $S$ with border plane $H$ and inner vector~$\vec v$.
  Then for every $s>0$ there exist finite sets
  \begin{alignat*}1
    &R^-\subseteq\set Z^r\cap\bigcup_{0\leq e\leq s}(H+e\vec v)\text{ and}\\
    &R^+\subseteq\set Z^r\cap\bigcup_{e>s}(H+e\vec v),
  \end{alignat*}
  and polynomials $b,b_{\vec i}\in\set K[\vec n]$ such that for any solution $y\in\set K(\vec n)$
  of \eqref{eq:main} we have
  \begin{equation}\label{eq:reduced}
    N^{\vec p}y=\frac{b+\sum_{\vec i\in R^+}b_{\vec i}N^{\vec i}y}
                    {\prod_{\vec i\in R^-}N^{\vec i-\vec p} a_{\vec p}}.
  \end{equation}
\end{lemma}

The sets $R^-$ and $R^+$ and the polynomials $b,b_{\vec i}$ can be computed for a
given $s$, $S$, $\vec p$, and~$\vec v$ by Algorithm~2 from~\cite{kauers10b}. The next theorem provides
a denominator bound with respect to~$W$. It is an adaption of Theorem~4 from~\cite{kauers10b}
to the present situation. We continue to assume the normalization $V=\set Z^t\times\{0\}^{r-t}$,
$W=\{0\}^t\times\set Z^{r-t}$.

\begin{theorem}\label{thm:db}
  Let $s\in\set N\cup\{-\infty\}$ be such that for any solution $y=p/q\in\set K(\vec n)$
  of~\eqref{eq:main} and any irreducible factors $u,v$ of $q$ with $\spread(u),\spread(v)\subseteq W$
  we have $\disp_1(u,v)\leq s$.
  Let $\vec p$ be a corner point of $S$ for which there is an inner vector $\vec v=(v_1,\dots,v_r)$
  with $v_1\geq1$ as well as an inner vector $\vec v'$ orthogonal to~$W$.
  For these choices of $s$, $\vec p$, and~$\vec v$, let $R^-$, $R^+$, $b$,~$b_{\vec i}$ be as in
  Lemma~\ref{lemma:1}.
  Let $a'_{\vec p}$ be the polynomial consisting of all the factors of $a_{\vec p}$
  whose spread is contained in~$W$.
  Then
  \begin{equation}\label{Equ:DenBoundPlance}
    d:=\prod_{\vec s\in R^-}N^{\vec s-2\vec p}a'_{\vec p}
  \end{equation}
  is a denominator bound of~\eqref{eq:main} with respect to~$W$.
\end{theorem}

\begin{proof}
  Let $y=p/q\in\set K(\vec n)$ be a solution of~\eqref{eq:main} and let $u$ be an irreducible
  factor of $q$ with multiplicity $m$ and $\spread(u)\subseteq W$.
  We have to show $u^m\mid d$.
  Lemma~\ref{lem:2} applied to $\vec p$ and $\vec v'$ implies that there is some $\vec i\in\set Z^r$
  with $u'\mid q$ and $u':=N^{\vec i}u\mid a_{\vec p}$.
  By the choice of~$s$ we have $\disp_1(u',u)\leq s$.

  Lemma~\ref{lemma:1} implies the representation
  \[
    N^{\vec p}y=\frac{b+\sum_{\vec i\in R^+}b_{\vec i}N^{\vec i}y}
                    {\prod_{\vec i\in R^-}N^{\vec i-\vec p} a_{\vec p}}.
  \]
  Because of $v_1>1$, every $\vec i\in R^+$ differs from $\vec p$ in the first coordinate
  by more than~$s$. This implies that $N^{\vec p}u$ and hence that $N^{\vec p}u^m$ cannot appear in the denominator of
  $N^{\vec i}y$ for any $\vec i\in R^+$. But it does appear in the denominator of~$N^{\vec p}y$,
  so it must appear as well in the denominator of the right hand side.
  The only remaining possibility is thus
  $N^{\vec p}u^m\mid \prod_{\vec i\in R^-}N^{\vec i-\vec p} a_{\vec p}$,
  and hence
  \[
    u^m\mid\prod_{\vec i\in R^-}N^{\vec i-2\vec p}a_{\vec p}.
  \]
  Because of $\spread(u')=\spread(u)\subseteq W$, it follows that $u^m\mid d$.
\end{proof}

The following figure illustrates the situation. The vector $\vec v$ is orthogonal to $H$
but not necessarily to~$W$, while the vector $\vec v'$ is orthogonal to $W$ but not
necessarily to~$H$.
Relation~\eqref{eq:reduced} separates $\vec p$ from the points in $R^+$ which are all below
the plane $H+s\vec v$. The points in $R^-$ are all between $H$ and $H+s\vec v$.

\medskip
\centerline{\epsfig{file=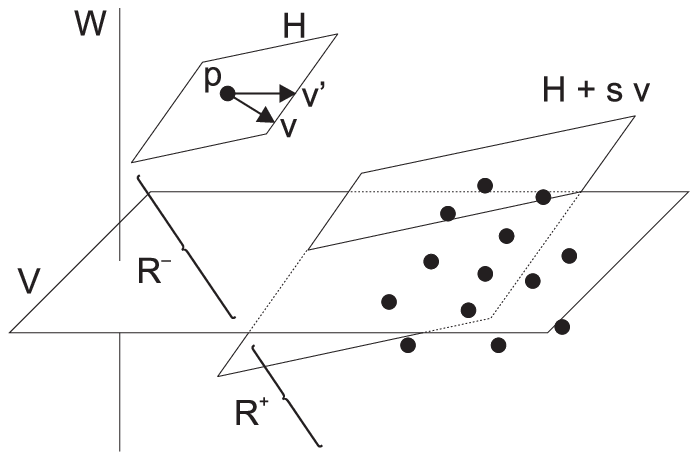}}

\medskip

\subsection{A Denominator Bounding Algorithm}

We now combine Theorems~\ref{thm:disp} and~\ref{thm:db} to an algorithm
for computing a denominator bound with respect to an arbitrary given~$W$
in situations where these theorems are applicable.

\begin{definition} Let $\vec p,\vec p'$ be corner points of~$S$
  and $W$ some submodule of~$\set Z^r$.
  \begin{enumerate}
  \item The point $\vec p$ is called \emph{useless} for $W$ if there is an edge $(\vec p,\vec s)$
    in the convex hull of $S\setminus\{\vec p\}$ with $\vec p-\vec s\in W$.
  \item The pair $(\vec p,\vec p')$ is called \emph{opposite} if there is a vector $\vec v$
    such that $(\vec s-\vec p)\cdot\vec v\geq0$ and
    $(\vec p'-\vec s)\cdot\vec v\geq0$ for all $\vec s\in S$.
    Such a $\vec v$ is called a \emph{witness vector} for the pair $(\vec p,\vec p')$.
  \item The pair $(\vec p,\vec p')$ is called \emph{useful} for $W$ if it is opposite and neither
    $\vec p$ nor $\vec p'$ is useless for~$W$.
  \end{enumerate}
\end{definition}

The definition of a useful pair is made in such a way that when a change of variables as
described in Section~\ref{sec:normal} is applied which maps a witness vector of the pair
to the first axis, then the sets $A$ and $B$ from Theorem~\ref{lem:2} are such that
$\vec p\in A$, $\vec p'\in B$ (because of the oppositeness), no two elements of $A$
agree in the first $r-\dim W$ coordinates (because $\vec p$ is not useless), and the
same is true for~$B$ (because $\vec p'$ is not useless).

Whether a pair $(\vec p,\vec p')\in S^2$ is useful or not can be found out by making an
ansatz for the coefficients of a witness vector and solving the system of linear
inequalities from the definition. The pair is useful if and only if this system is
solvable, and in this case, any solution gives rise to a witness vector.

If for a given submodule $W$ we have found a useful pair, then we can compute a denominator
bound with respect to $W$ by the following algorithm.

\begin{algorithm}\label{algo:1}
\emph{Input:} An equation of the form \eqref{eq:main}, a submodule $W$ of~$\set Z^r$,
   a useful pair $(\vec p,\vec p')$ of~$S$ for~$W$.
\emph{Output:} A denominator bound for \eqref{eq:main} with respect to~$W$.

\begin{algo}%
  \text{Set $t:=r-\dim W$.}
  \text{Choose $\vec v_1,\vec v_2,\dots,\vec v_t\in\set Z^r$ such that $\vec v_1$ is a witness}\label{alg:1:2}
  ^^I^^I^^I\text{vector for $(\vec p,\vec p^\prime)$ and $\set Z^r=V\oplus W$ where}
  ^^I^^I^^I\text{$V$ is the module generated by these vectors.}
  \text{Perform a change of variables as described in}
  ^^I^^I^^I\text{Section~\ref{sec:normal} such that $\vec v_i$ becomes the $i$th unit}
  ^^I^^I^^I\text{vector in $\set Z^r$, $W$ becomes $\{0\}^t\times\set Z^{r-t}$.}
  \text{Determine $A,B$ as in Theorem~\ref{lem:2}.}\label{a:1:11}
  \text{Compute $s\in\set N\cup\{-\infty\}$ as defined in Theorem~\ref{lem:2}.}
  \text{Choose an inner vector $\vec v\in\set R^r$ for $\vec p$.}
  \text{Compute $R^-$ as defined in Lemma~\ref{lemma:1}.}
  \text{Compute $d$ as defined in Theorem~\ref{thm:db}.}
  \text{Apply the inverse change of variables to~$d$, getting $d^\prime$.}
  \text{Return $d^\prime$.}%
\end{algo}
\end{algorithm}

The following variations can be applied for further improvements:

\begin{enumerate}
\item If the dimension of $V$ is larger than 1, there might be different choices
  of witness vectors. Choosing different versions in line~\ref{alg:1:2} might
  lead to different denominator bounds of $W$, say, $d_1,\dots,d_k$. Then taking
  $d:=\gcd(d_1,\dots,d_k)$ leads to a sharper denominator bound
  for~\eqref{eq:main} w.r.t.\ $W$.

\item Choosing different inner vectors in line~10 might lead to different sets $R^-$
to write~\eqref{eq:reduced} and hence gives rise to different denominator bounds in~\eqref{Equ:DenBoundPlance}. Taking the gcd of these denominator bounds produces a refined version.
\end{enumerate}

We remark that the coefficients $a_{\vec{s}}$ with $\vec{s}\in S$ are often available in factorized form. Then also the denominator bounds are obtained in factorized form, and the gcd-computations reduce to comparisons of these factors and bookkeeping of their multiplicities.

\section{A Combined Denominator Bound}

As mentioned earlier, when setting $W=\{0\}$, one is able to derive an aperiodic
denominator bound for equation~\eqref{eq:main}. In this particular case, for
each corner point $\vec p$ there is an other corner point $\vec p'$ such that
$(\vec p,\vec p')$ is useful for $W$. Hence applying Algorithm~\ref{algo:1} for
any useful pair leads to an aperiodic denominator bound. In particular, running
through all corner points and taking the gcd for all these candidates leads to a
rather sharp aperiodic denominator bound for equation~\eqref{eq:main} which
coincides with the output given in our previous investigation~\cite{kauers10b}.

In the other extreme, when setting $W=\set Z^r$, a denominator bound
for~\eqref{eq:main} w.r.t.\ $W$ would lead to a complete denominator bound for
equation~\eqref{eq:main}. However, in this case, we will fail to find 
a useful pair $(\vec p,\vec p')$, and our Algorithm~\ref{algo:1} is not
applicable.

Our goal is to find a simultaneous denominator bound with respect to all $W$
to which Algorithm~\ref{algo:1} is applicable, i.e., for all $W$ from the set
$$
 U:=\{W\text{ submodule of }\set Z^r\mid \exists\ (\vec p,\vec p') \text{ useful for }W\}.
$$
In general, this is an infinite set. But we can make use of the observations
made after Example~\ref{ex:2}. Using Lemma~\ref{lem:2}, it turns out that instead 
of looping through all these infinitely many modules~$W$, it is sufficient to
consider those $W$ which appear as spread of some factor in the coefficient of
$a_{\vec p}$.

This argument even works for all $W$ in the larger set
\begin{alignat*}1
O:=\{&W\text{ submodule of }\set Z^r\mid\\
&\quad \exists\ (\vec p,\vec p') \text{ opposite with $\vec p$ not useless for $W$}\},
\end{alignat*}
but since the $W\in O\setminus U$ do not satisfy the conditions of Theorem~\ref{thm:disp},
we can only obtain partial information about their denominator bounds. 

We propose the following algorithm. 


\begin{algorithm}\label{algo:2}
\emph{Input:} An equation of the form \eqref{eq:main}.
\emph{Output:} A finite set of irreducible polynomials $P=\{p_1,\dots,p_k\}$, and
  a nonzero $d\in\set K[\vec n]$ such that for every solution $y=\frac{p}{q}\in\set K(\vec n)$ of~\eqref{eq:main} and every irreducible factor $u$ of $q$ with multiplicity $m$ exactly one of the following holds:\\
  1. $\spread(u)\in U$ and $u^m\mid d$,\\
  2. $\spread(u)\in O\setminus U$ and $\exists\ \vec s\in\set Z^r, p\in P: N^{\vec s}u=p$,\\
  3. $\spread(u)\notin O$.

\begin{algo}%
  d:=1
  P:=\{\}
  C:=\{\vec p\in S: \vec p\text{ is a corner point of $S$}\}
  |forall| \vec q\in C |do|
  ^^I |forall| u\mid a_{\vec q}\text{ irreducible} |do|\label{alg:2:5}
  ^^I^^I W := \spread(u)
  ^^I^^I |if| W\in U |then|
  ^^I^^I^^I \text{Compute a denominator bound $d_0$ w.r.t.\ $W$}\label{alg:2:9}
  ^^I^^I^^I^^I^^I^^I \text{using an arbitrary useful pair for $W$.}
  ^^I^^I^^I d := \lcm(d,d_0)\label{alg:2:11}
  ^^I^^I |else\ if| W\in O |then|
  ^^I^^I^^I P := P \cup\{u\}\label{alg:2:12}
  |return| (P,d)%
\end{algo}
\end{algorithm}

\begin{theorem}\label{thm:combined}
  The polynomial $d$ computed by Algorithm~\ref{algo:2}
  is a denominator bound with respect to any finite union of modules in~$U$.
\end{theorem}
\begin{proof}
  Let $W$ be in~$U$ and $(\vec p,\vec p')$ be a useful pair with respect to~$W$.
  Let $y=p/q$ be a solution of~\eqref{eq:main} and $u$ be an irreducible factor
  of $q$ with multiplicity~$m$ and $\spread(u)\subseteq W$.
  We have to show that $u^m\mid d$.

  Since $\vec p$ is not useless, Lemma~\ref{lem:2} implies that there is some
  $\vec i\in\set Z^r$ with $N^\vec i u\mid a_{\vec p}$. This factor is going
  to be investigated in some iteration of the loop starting in line~\ref{alg:2:5}.
  The polynomial $d_0$ computed in this iteration is a denominator bound with
  respect to $\spread(N^\vec i u)=\spread(u)$ and hence a forteriori a 
  denominator bound with respect to~$W$. 
  It follows that $u^m\mid d_0\mid d$.

  This proves the theorem when $W$ itself is in~$U$. If $W$ is only a finite union
  of elements of~$U$, the theorem follows from here by Lemma~\ref{lemma:lcm}.
\end{proof}

For $W\in O\setminus U$, we can still apply Lemma~\ref{lem:2} but
Theorem~\ref{thm:disp} is no longer applicable. This prevents us from computing
precise denominator bounds with respect to these~$W$. However, using the set
$P=\{p_1,\dots,p_k\}$ returned by the algorithm we can at least say that 
for every denominator $q$ of a solution $y=p/q$ of~\eqref{eq:main} there
exist $m\in\set N$ and a finite set $S'\subseteq\set Z^r$ such that
\begin{equation}
   d\,\prod_{\substack{p\in P\\ \vec s\in S'}}N^{\vec s}p^m
\end{equation}
is a multiple of every divisor of $q$ whose spread is contained in some finite
union of modules in~$O$. Appropriate choices $S'$ and $m$ can be found for instance by making an
ansatz. Note also that the set~$P$ is usually smaller than the set of all
periodic factors that occur in the coefficients $a_{\vec s}$
of~\eqref{eq:main}. This phenomenon was demonstrated already in the second part
of Example~\ref{ex:2}.


Summarizing, some part of the denominator is out of reach, namely all those parts of the denominator w.r.t.\ the modules from
\begin{align*}
\{&W\text{ submodule of }\set Z^r\mid\\
&\quad \forall\ (\vec p,\vec p') \text{ opposite for $W$ with $\vec p$ and $\vec p'$ useless for $W$}\},
\end{align*}
some part of the denominator can be given up to possible shifts and multiplicities, and a big part of our denominator bound can be given explicitly by~$d$.

The following improvements can be utilized.

\begin{enumerate}
\item As preprocessing step, one should compute an aperiodic denominator bound for the equation~\eqref{eq:main} as described above. What remains is to recover the periodic factors. As a consequence, one can neglect all irreducible factors $u$ which are aperiodic and one can apply Theorem~\ref{thm:db} where all aperiodic factors are removed from the polynomials $a'_{\vec{p}}$.

\item Choosing different useful pairs for a module $W$ in line~9 might lead to different choices of denominator bounds, and taking their gcd gives rise to sharper denominator bounds of~\eqref{eq:main} w.r.t.\ $W$.
\end{enumerate}

\section{Discussion}

Typically the set $O$ will contain all the submodules of $\set Z^r$. Only when
the convex hull of $S$ happens to have two parallel edges on opposite sides, as
is the case in Example~\ref{ex:2}, then modules $W$ parallel to this edge do not
belong to~$O$.  The set $U$ will never contain all the submodules of~$\set
Z^r$. Precisely those modules $W$ which are parallel to an edge of the convex
hull of~$S$ do not belong to~$U$. Since the convex hull of $S$ contains only
finitely many edges, $U$~will in some sense still contain almost all the
submodules of~$\set Z^r$.

Depending on the origin of the equation, it may be that there is some freedom in
the structure set~$S$. For example, by multivariate guessing~\cite{kauers09a} or
by creative telescoping~\cite{zeilberger91,chyzak00,schneider04c} one can
systematically search for equations with a prescribed structure set. In such
situations, one can try to search for an equation with a structure set for which
$U$ and~$O$ cover as many spaces as possible.

If two equations with different structure sets are available, it may be possible
to combine the two denominator bounds obtained by Algorithm~\ref{algo:2} to a
denominator bound with respect to the full space $\set Z^r$.

\begin{example}
  Consider the following system of equations:
  \begin{alignat*}1
    &-(k+n+1) (2 k+3 n+1) y(n,k)\\
    &\quad+(k+n+4) (2 k+3 n+3) y(n,k+1)\\
    &\qquad-(k+n+2) (2 k+3 n+4)y(n+1,k)\\
    &\qquad\quad+(k+n+5) (2 k+3 n+6) y(n+1,k+1) =0 ,\\
    &(n^2+n+1) (2 k+3 n+3) y(n,k+1)\\
    &\quad-(n^2+5 n+7) (2 k+3 n+4)y(n+1,k)\\
    &\qquad-(n^2+3 n+3) (2 k+3 n+8) y(n+1,k+2)\\
    &\qquad\quad+(n^2+7 n+13) (2 k+3 n+9) y(n+2,k+1).
  \end{alignat*}
  Algorithm~\ref{algo:2} applied to the first equation returns
  \[
    d=(n+k+1)(n+k+2)(n+k+3)(3n+2k+1)
  \]
  as a denominator bound with respect to any $W$ except $\binom10\set Z$ and
  $\binom01\set Z$. Applied to the second equation, it returns
  \[
    d=(n^2+n+1)((n+1)^2+(n+1)+1)(3n+2k+1)
  \]
  as a denominator bound with respect to any $W$ except $\binom11\set Z$ and
  $\binom1{-1}\set Z$. The least common multiple of the two outputs is a
  simultaneous denominator bound with respect to any~$W$.

  Indeed, the system has the solution
  \[
    \tfrac1{(n+k+1)(n+k+2)(n+k+3)(n^2+n+1)((n+1)^2+(n+1)+1)(3n+2k+1)}.
  \]
\end{example}

There is no hope for an algorithm which computes for any given single equation a
denominator bound with respect to the full space~$\set Z^r$. This is because
there are quations whose solution space contains rational functions with no
finite common denominator. For instance, for every univariate polynomial~$p$, we
have that $1/p(n+k)$ is a solution of the equation
\[
  y(n+1,k)-y(n,k+1)=0.
\]
It would be interesting to characterize under which circumstances this happens,
and to have an algorithm which finds a denominator bound with respect to $\set Z^r$
in all other cases. 

\bibliographystyle{plain}
\bibliography{all}

\end{document}